\documentclass[journal]{IEEEtran}

\usepackage{amsmath}
\usepackage{amssymb}
\usepackage{amsfonts}
\usepackage{amsthm}
\usepackage{esint}
\usepackage[hidelinks]{hyperref}
\usepackage{empheq}
\usepackage{cite}
\usepackage{xcolor}
\usepackage{dsfont} 
\usepackage[mathscr]{euscript} 
\usepackage{enumerate}
\usepackage{scrextend}
\usepackage{centernot}
\usepackage[shortlabels]{enumitem}
\usepackage{xfrac}
\usepackage[normalem]{ulem}

\usepackage{tikz}
\usepackage{graphicx}
\usepackage{color}
\usetikzlibrary{intersections}
\usetikzlibrary{calc}

\ifCLASSOPTIONcompsoc
\usepackage[caption=false,font=normalsize,labelfont=sf,textfont=sf]{subfig}
\else
\usepackage[caption=false,font=footnotesize]{subfig}
\fi

\newtheorem{theo}{Theorem}
\newtheorem{lem}{Lemma}
\newtheorem{cor}{Corollary}

\theoremstyle{definition}
\newtheorem{defi}{Definition}

\theoremstyle{remark}

\newcommand{\bos}[1]{\boldsymbol{#1}}

\newcommand{\da}{\downarrow}
\newcommand{\pr}{\mathcal{P}^\da_d}

\DeclarePairedDelimiterX{\infdivx}[2]{(}{)}{%
  #1\;\delimsize\|\;#2%
}
\newcommand{\rel}{D \infdivx}

\begin{document}

\title{Majorization precursors to supermodularity and subadditivity on the majorization lattice}

\DeclareRobustCommand{\IEEEauthorrefmark}[1]{\smash{\textsuperscript{\footnotesize #1}}}
\author{\vspace{10pt}
    \IEEEauthorblockN{Alexander St\'evins\IEEEauthorrefmark{1}, Michael G. Jabbour\IEEEauthorrefmark{2,1}, Serge Deside\IEEEauthorrefmark{1}, and Nicolas J. Cerf\IEEEauthorrefmark{1}~\\~\\}
    \IEEEauthorblockA{\IEEEauthorrefmark{1}{\small Centre for Quantum Information and Communication, \'Ecole polytechnique de Bruxelles,~\\ CP 165/59, Universit\'e libre de Bruxelles, 1050 Brussels, Belgium}~\\}
    \IEEEauthorblockA{\IEEEauthorrefmark{2}{\small SAMOVAR, Télécom SudParis, Institut Polytechnique de Paris, 91120 Palaiseau, France}}
}

\maketitle

\begin{abstract}
    We establish two structural majorization relations, which we call \textit{precursors}, underlying the properties of supermodularity and subadditivity on the lattice induced by majorization.
    These are precursors in that  they immediately imply that all sums of concave functions, which we dub \textit{sum-concave} functions, are supermodular and subadditive on the majorization lattice.
    Using these majorization relations, we then show the supermodularity and subadditivity (in the lattice-theoretic sense) of Tsallis entropies (for all $\alpha$) and R\'enyi entropies (for all $\alpha > 1$), also recovering these properties for the Shannon entropy in the process.
    We further strengthen these inequalities, showing that: (i) all these entropic functionals are \emph{strictly} subadditive on the majorization lattice; (ii) Tsallis entropies (and therefore the Shannon entropy as well) are \emph{strictly} supermodular on the majorization lattice.
\end{abstract}

\begin{IEEEkeywords}
Majorization lattice, Shannon entropy, R\'enyi entropy, Tsallis entropy, supermodularity, subadditivity.
\end{IEEEkeywords}

\section{Introduction}

Majorization theory is a rich area of mathematics which finds applications in a wide variety of fields, extending in particular to information theory~\cite{Majorization}. The majorization relation is a pre-order on the set of probability mass functions, and is thus intimately connected with various measures of ``uncertainty" or ``disorder" in information theory. Notably, being a paradigmatic measure of uncertainty, the Shannon entropy enjoys such a connection:
a majorization relation between two probability mass functions necessarily implies an inequality on their Shannon entropies. A similar connection holds for other entropic functions, such as the Tsallis or Rényi entropies. Majorization theory can thus been used as a powerful tool to prove inequalities in information theory, for instance continuity bounds for entropies, see e.g.~\cite{HansonDatta2022Entropies,Jabbour2022,Jabbour2024}.
In the last two decades, majorization has been shown to induce a lattice structure~\cite{Cicalese2002}, which guarantees the existence and uniqueness of a greatest lower bound (i.e., the \textit{meet} $\wedge$) and a least upper bound (i.e., the \textit{join} $\vee$) of any two probability mass functions for the majorization order. In this sense, the lattice provides a natural extended framework for majorization theory, allowing entropic functions to exhibit new lattice-theoretic properties. For instance, the Shannon entropy was shown to be supermodular and subadditive (in the lattice-theoretic sense) on the majorization lattice~\cite{Cicalese2002}.

In this paper, we exploit the formalism of the majorization lattice to show that a larger family of entropies satisfies the properties of subadditivity and supermodularity on the majorization lattice. We first focus on the family of functions we dub \textit{sum-concave} functions, which are defined as sums of concave functions. Such a family includes Tsallis entropies (and in particular the Shannon entropies). By studying the subadditivity and supermodularity of \textit{sum-concave} functions on the majorization lattice, we are moreover able to show such properties for R\'enyi entropies. In fact, we show that the properties can be strengthened in various cases using a Kullback–Leibler divergence.

\medskip

\noindent
{\bf{Layout of the paper and summary of our contributions:}} Relevant notations and definitions are given in Section~\ref{sec:pre}.
In Section~\ref{sec:precursors}, our main results are given, consisting of two structural majorization relations -- which we call \textit{majorization precursors} -- on the majorization lattice. We prove that for any two distributions $\bos{p}, \bos{q} \in \pr$, the structural majorization relations $(\bos{p} \wedge \bos{q}) \oplus (\bos{p} \vee \bos{q}) \prec \bos{p} \oplus \bos{q}$ (Theorem~\ref{th:concatenation}) as well as $\bos{p} \oplus \bos{q} \prec (\bos{p} \wedge \bos{q}) \oplus (1, 0, \dots, 0)$ (Theorem~\ref{th:subadditive_concatenation}) hold, which yields upper and lower majorization bounds on the direct sum  $\bos{p} \oplus \bos{q}$. These two results directly imply the supermodularity and subadditivity (in the lattice-theoretic sense) of a family of functions (which we call \textit{sum-concave} functions), greatly simplifying the proofs of such properties for various entropic functions.
In Section~\ref{sec:entropies}, we demonstrate the usefulness of these majorization precursors by proving the supermodularity and subadditivity of the  Tsallis entropies on the majorization lattice, recovering these properties for the Shannon entropy as a special case, and prove that these inequalities can even be further tightened by providing explicit correction terms. Moreover, we show the submodularity and log-submodularity of $\ell_p$ norms for $p > 1$ on the majorization lattice, which we use to prove the supermodularity (for $\alpha > 1$) and the subadditivity (for $\alpha \geq 0$) of R\'enyi entropies on the majorization lattice.
Finally, Section~\ref{sec:conclusion} contains a final discussion and some future research directions.

\section{Preliminaries\label{sec:pre}}

Given $\bos{p} \in \mathbb{R}^d$ for some finite dimension $d$, let $\bos{p}^{\da} \in \mathbb{R}^d$ be the vector containing the elements of $\bos p$ arranged in non-increasing order. For $\bos{p}, \bos{q} \in \mathbb{R}^d$, we say $\bos{p}$ is majorized by $\bos{q}$, (denoted as $\bos{p} \prec \bos{q}$)~\cite{Majorization}, if
\begin{equation} \label{eq:majorization}
    \sum_{j=1}^k p^{\da}_j \leq \sum_{j=1}^k q^{\da}_j, \; \forall k = 1, \dots, d-1, \; \text{and} \; \sum_{j=1}^d p^{\da}_j = \sum_{j=1}^d q^{\da}_j.
\end{equation}
Note that throughout, when the dimensions of $\bos p$ and $\bos q$ do not match, we append zeros to the shortest distribution until the dimensions match. The majorization relation is related to the notion of order in the sense that the relation $\bos{p} \prec \bos{q}$ can be understood as checking whether an element $\bos{q}$ concentrates more weight in its biggest elements than $\bos{p}$. When such a relation is satisfied, one can say that $\bos{p}$ is less certain than $\bos{q}$, or more ordered. For easiness of notation, we denote the cumulative sums as
\begin{equation}
    S^\da_k(\bos{p}) = \sum_{j=1}^k p^{\da}_j.
\end{equation}
Moreover, we also focus on probability vectors, i.e. vectors from the set \begin{equation}
    \mathcal{P}_d \coloneqq \left\{\bos x \in \mathbb{R}^d | x_i \geq 0 , \sum_i x_i = 1\right\}.
\end{equation}
The majorization relation does not constitute a total order on probability vectors but only a preorder, meaning that there can be pairs of vectors $\bos{p}, \bos{q} \in \mathcal{P}_d$ such that
neither $\bos{p} \prec \bos{q}$ nor $\bos{q} \prec \bos{p}$ is true.
Such pairs of vectors are said to be \emph{incomparable} (denoted $\bos{p} \nsim \bos{q}$). Instead, when at least $\bos{p} \prec \bos{q}$ or $\bos{q} \prec \bos{p}$ is true, we say that the pair of vectors is \emph{comparable} (denoted $\bos p \sim \bos q$). When restricted to the set of probability vectors rearranged in non-increasing order defined as
\begin{equation}
    \pr \coloneqq \left\lbrace \bos{p} \in \mathcal{P}_d | p_k \geq p_{k+1}, \forall k = 1, \dots, d-1 \right\rbrace,
\end{equation}
the majorization preorder becomes a partial ordering on $\pr$ (since the conditions $\bos{x} \prec \bos{y}$ and $\bos{y} \prec \bos{x}$ together imply $\bos{x} = \bos{y}$ in that case), which is thus a partially ordered set, or poset for short~\cite{Davey2002}.

Now, a function $f:A \subseteq \mathbb{R}^d \to \mathbb{R}$ is said to be Schur-convex if $\bos{p} \prec \bos{q}$ implies $f(\bos{p}) \leq f(\bos{q})$ for all $\bos{p}, \bos{q} \in A$, and it is said to be Schur-concave if $(-f)$ is Schur-convex. The class of functions of the form
\begin{equation} \label{eq:sum-cc}
    F(\bos{x}) = \sum_{i} \varphi(x_i),
\end{equation}
where $\varphi$ is either convex or concave,
appears at several points in this work, and so depending on the convexity/concavity of $\varphi$, we call $F$ \textit{sum-convex} and \textit{sum-concave}, respectively.

The following Lemma is a fundamental result in majorization theory, showing that what we call sum-convex (resp. sum-concave) functions are all Schur-convex (resp. Schur-concave).
\begin{lem}[Schur, Hardy, Littlewood and P\'olya~\cite{Majorization}] \label{lem:karamata}
    Let $f$ be a real-valued convex function on $[0, 1]$. If $\bos{p} \prec \bos{q} \in \mathcal{P}_d$, then
    \begin{equation}
        \sum_{i=1}^{d} f(p_i) \leq \sum_{i=1}^{d} f(q_i).
    \end{equation}
    This can be restated as saying that the function $F = \sum f$ is Schur-convex.
\end{lem}

Since $f$ is concave if $-f$ is convex, one can also use Lemma~\ref{lem:karamata} on concave functions, the inequality simply being reversed. Notable examples of Schur-concave functions are entropies, which is expected since entropies, just like majorization, can be interpreted as tools to compare the uncertainty of probability distributions. The most common entropy in information theory is the Shannon entropy, introduced by Shannon in Ref.~\cite{Shannon} as a fundamental quantity for his famous data compression and noisy coding theorems.
The Shannon entropy of a probability distribution $\bos{p} \in \mathcal{P}_d$ is defined as
\begin{equation}
    H(\bos{p}) \coloneqq - \sum_{i=1}^d p_i \log p_i,
\end{equation}
whose units are \textit{bits} if the logarithm is taken in base 2, and \textit{nats} when the logarithm is taken in base $e$. We take all logarithms to base 2 in this work.
A comprehensive overview of the Shannon entropy and its applications can be found in Ref.~\cite{CoverThomas}. The Shannon entropy is clearly sum-concave, as can be seen by taking $\varphi(x) = -x \log x$ in the definition, and is thus Schur-concave by Lemma~\ref{lem:karamata}. A related quantity is the Kullback-Leibler divergence (or relative entropy) $D$, defined as
\begin{equation}
  \rel{\bos{p}}{\bos{q}} \coloneqq   \sum_{i:p_i>0} p_i \log \frac{p_i}{q_i},
\end{equation}
whenever $\operatorname{supp}(\bos p)\subseteq \operatorname{supp}(\bos q)$. If this support condition is not satisfied, we set $\rel{\bos p}{\bos q}:=+\infty$. In the cases considered below, this support condition is always satisfied.

Several generalizations of the Shannon entropy have been proposed over the years, such as the so-called R\'enyi entropies $H_\alpha$~\cite{Renyi1961}. which are widely used in physics and quantum information to characterize entanglement catalysis~\cite{Turgut2007, Klimesh2007, Aubrun2008}.
To introduce them, for any vector $\bos x \in \mathbb{R}^d_+$ we define
\begin{equation} \label{eq:ell_p}
    ||\bos{x}||_p \coloneqq \left(\sum_{i=1}^d x_i^p\right)^{\frac{1}{p}}.
\end{equation}
For $p \geq 1$, Eq.~(\ref{eq:ell_p}) defines the $\ell_p$ norm of the vector $\bos x$. For $p < 1$, Eq.~(\ref{eq:ell_p}) defines a quasinorm instead as it fails the triangle inequality. The R\'enyi entropy of order $\alpha \in \mathbb{R}^+ \backslash \{1\}$ of a probability distribution $\bos p \in \mathcal{P}_d$ is then defined as

\begin{equation}
    H_\alpha(\bos{p}) \coloneqq \frac{1}{1-\alpha}\log(||\bos{p}||_\alpha^\alpha).
\end{equation}
\noindent The Shannon entropy can be recovered from the R\'enyi entropy by taking the limit as $\alpha \rightarrow 1$. 

More recently, the so-called Tsallis entropy $T_\alpha$ was introduced, as a non-additive generalization of the Gibbs-Boltzmann entropy in thermodynamics~\cite{Tsallis1988}, and has been exploited in various fields, such as entanglement detection~\cite{Canosa2002} or the study of non-extensive systems~\cite{Lyra1998}.
The Tsallis entropy of order $\alpha \in \mathbb{R}^+ \backslash \{1\}$ of a distribution $\bos p \in \mathcal{P}_d$ is defined as

\begin{equation}
    T_\alpha(\bos{p}) \coloneqq \frac{1}{1 - \alpha} \left(||\bos{p}||_\alpha^\alpha - 1\right).
\end{equation}
Again, the Shannon entropy can be recovered from the Tsallis entropy by taking the limit as $\alpha \rightarrow 1$. Tsallis entropies are sum-concave functions for all values of $\alpha$, as can be seen by choosing $\varphi(x) = \frac{1}{1-\alpha}(x^\alpha - x)$ in Eq.~(\ref{eq:sum-cc}).

As discussed previously, some pairs of vectors are incomparable, and in that case the usual majorization properties (e.g., Schur-convexity) cannot be used to deduce anything about the relationship between the two states. However, Cicalese and Vaccaro showed that the majorization pre-order induces a lattice structure~\cite{Cicalese2002}, which guarantees that when two states are incomparable, one can still find a well-defined greatest lower bound and least upper bound, which we can then compare to our states using usual majorization properties.
\begin{defi}[Majorization lattice~\cite{Cicalese2002}]
    The majorization lattice is a quadruple $\left< \mathcal{P}^\da_{d}, \prec, \wedge, \vee \right>$, where the two operations $\wedge$ and $\vee$ are defined as follows:
    \begin{enumerate}[(i)]
        \item The \emph{meet} of two elements $\bos{p}, \bos{q} \in \pr$, denoted $\bos{p} \wedge \bos{q}$, is the unique element in $\pr$ such that, for any $\bos{r} \in \pr$ satisfying both $\bos{r} \prec \bos{p}$ and $\bos{r} \prec \bos{q}$, we have\footnote{Intuitively, the meet is thus the \textit{greatest common majorized} distribution, since it majorizes all other distributions that are majorized by both $\bos{p}$ and $\bos{q}$, i.e. the greatest lower bound for the majorization partial order.} $\bos{r} \prec \bos{p} \wedge \bos{q}$.
        \item The \emph{join} of two elements $\bos{p}, \bos{q} \in \pr$, denoted $\bos{p} \vee \bos{q}$, is the unique element in $\pr$ such that, for any $\bos{r} \in \pr$ satisfying both $\bos{p} \prec \bos{r}$ and $\bos{q} \prec \bos{r}$, we have\footnote{Similarly, the join is thus the \textit{least common majorizer}, i.e. the least upper bound for the majorization partial order.} $\bos{p} \vee \bos{q} \prec \bos{r}$.
    \end{enumerate}
\end{defi}

A schematic representation of the majorization lattice can be found in Figure~\ref{fig:lattice}. Explicitly, the meet $\bos{p} \wedge \bos{q}$, which we denote by $\bos{m}$ throughout for easiness of notation, can be constructed as follows. We define $\bos{\alpha}(\bos{p}, \bos{q}) = (a_1, a_2, \dots, a_d)$ as
\begin{align}
    a_i &= \min\left\{\sum^i_{k=1} p_k, \sum^i_{k=1} q_k\right\} - \min\left\{\sum^{i-1}_{k=1} p_k, \sum^{i-1}_{k=1} q_k\right\}\label{eq:alpha}\\
        &= \min\left\{\sum^i_{k=1} p_k, \sum^i_{k=1} q_k\right\} - a_{i-1}, \label{eq:alpha_bis}
\end{align}
then $\bos{p} \wedge \bos{q} = \bos{\alpha}(\bos{p}, \bos{q})$. For the join $\bos{p} \vee \bos{q}$, which we denote by $\bos{j}$ throughout, one proceeds similarly.  We define $\bos{\beta}(\bos{p}, \bos{q}) = (b_1, b_2, \dots, b_d)$ as
\begin{align}
    b_i &= \max\left\{\sum^i_{k=1} p_k, \sum^i_{k=1} q_k\right\} - \max\left\{\sum^{i-1}_{k=1} p_k, \sum^{i-1}_{k=1} q_k\right\}\label{eq:beta}\\
        &= \max\left\{\sum^i_{k=1} p_k, \sum^i_{k=1} q_k\right\} - b_{i-1}. \label{eq:beta_bis}
\end{align}
It is not guaranteed at this stage that $\bos{\beta}(\bos{p}, \bos{q})$ is sorted in non-increasing order, so one needs to smooth the convex dents in the Lorenz curve to obtain the true join
. For the purposes of this article it is enough to know that $\bos{p} \vee \bos{q} \prec \bos{\beta}^\da(\bos{p}, \bos{q})$. An explicit construction is detailed in Ref.~\cite{Cicalese2002}.

We now list some useful properties of functions on the majorization lattice.
\begin{defi}[Supermodularity]
    A function $\varphi : \mathcal{L} \rightarrow \mathbb{R}$ is supermodular over a lattice $\left< \mathcal{L}, \prec, \wedge, \vee \right>$ if and only if, for all $\bos{a}, \bos{b} \in \mathcal{L},$
    \begin{equation}
        \varphi(\bos{a} \wedge \bos{b}) + \varphi(\bos{a} \vee \bos{b}) \geq \varphi(\bos{a}) + \varphi(\bos{b}).
    \end{equation}
    Moreover, $\varphi$ is \textit{strictly} supermodular if the inequality is replaced by a strict inequality for all incomparable $\bos{a}, \bos{b} \in \mathcal{L}$.
\end{defi}
\noindent Similarly, a function $\varphi$ is (strictly) submodular on a lattice if $-\varphi$ is (strictly) supermodular.
\begin{defi}[Subadditivity]
    A function $\varphi : \mathcal{L} \rightarrow \mathbb{R}$ is subadditive over a lattice $\left< \mathcal{L}, \prec, \wedge, \vee \right>$ if and only if, for all $\bos{a}, \bos{b} \in \mathcal{L},$
    \begin{equation}
        \varphi(\bos{a} \wedge \bos{b}) \leq \varphi(\bos{a}) + \varphi(\bos{b}).
    \end{equation}
    Moreover, $\varphi$ is \textit{strictly} subadditive if the inequality is replaced by a strict inequality for all incomparable $\bos{a}, \bos{b} \in \mathcal{L}$.
\end{defi}
\noindent Similarly, a function $\varphi$ is (strictly) superadditive if $-\varphi$ is (strictly) subadditive.

\begin{figure}[t!]
    \centering
    \begin{tikzpicture}[scale=0.9]
        \coordinate (A) at (-2,3);
        \coordinate (B) at (2, -3);
        
        \coordinate (C) at (-2,-3);
        \coordinate (D) at (2,3);
        \path [name path=A--B] (A) -- (B);
        \path [name path=C--D] (C) -- (D);
        \path [name intersections={of=A--B and C--D,by=E}];

        \coordinate (F) at (-0.666,3);
        \coordinate (G) at (3.333, -3);

        \coordinate (H) at (0.666,-3);
        \coordinate (I) at (4.666,3);
        \draw [name path=F--G] (F) -- (G);
        \draw [name path=H--I] (H) -- (I);
        \path [name intersections={of=F--G and H--I,by=J}];

        \path [name intersections={of=F--G and C--D,by=K}];
        \path [name intersections={of=H--I and A--B,by=L}];

        \coordinate (M) at (-2.5, -2);
        \coordinate (N) at (-2.5, 2);
        \draw[->, line width = 0.3mm] (M) -- (N);
        \node [inner sep=0pt, label=0:$F$] at (-2.5, 0) {};

        \fill[top color=transparent!0, bottom color=green!10] (A) -- (D) -- (E) -- cycle;
        \fill[top color=transparent!0, bottom color=green!10] (F) -- (I) -- (J) -- cycle;
        \fill[top color=blue!10, bottom color=transparent!0] (B) -- (C) -- (E) -- cycle;
        \fill[top color=blue!10, bottom color=transparent!0] (G) -- (H) -- (J) -- cycle;

        \draw [name path=A--B, color=blue] (A) -- (B);
        \draw [name path=C--D, color=blue] (C) -- (D);
        \draw [name path=F--G, color=orange] (F) -- (G);
        \draw [name path=H--I, color=orange] (H) -- (I);

        \node [fill=black,inner sep=1pt,label=0:$\bos{p}$] at (E) {};
        \node [fill=black,inner sep=1pt,label=0:$\bos{q}$] at (J) {};
        \node [fill=black,inner sep=1pt,label=0:$\bos{p} \wedge \bos{q}$] at (K) {};
        \node [fill=black,inner sep=1pt,label=180:$\bos{p} \vee \bos{q}$] at (L) {};
        
    \end{tikzpicture}
    \caption{Schematic representation of the majorization lattice arising from two probability distributions, $\bos{p}$ and $\bos{q}$, and their majorization cones. The green regions, which we call the \textit{past cones} are the sets of states that are majorized by one of the two distributions, whereas the blue regions, which we call the \textit{future cones} are the sets of states that majorize one of the two distributions. It is interesting to note that the tip of the intersection of two past cones gives the meet of the two distributions generating the cones, and similarly for future cones and the join. The arrow indicates the direction in which any Schur-concave function $F$ (such as entropy) increases. Alternative representations can be found in Ref.~\cite{deOliveiraJunior2022}.}
    \label{fig:lattice}
\end{figure}
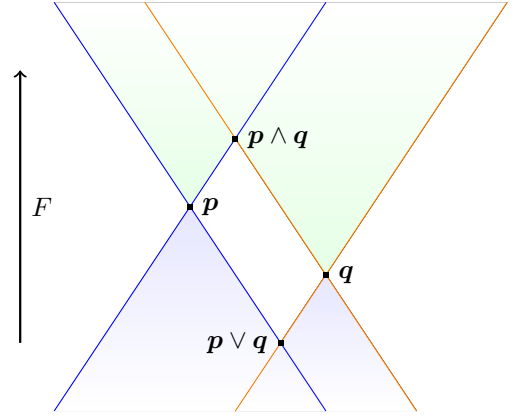

Supermodularity in particular has been studied extensively in the economics literature, as the supermodularity of a cost function implies the existence of optimal solutions to optimization problems~\cite{Topkis1998}.
In quantum information, majorization plays a significant role, as the link between single-copy interconvertibility of entangled pure states and a majorization relation of their Schmidt vectors was discovered about twenty-five years ago by Nielsen~\cite{Nielsen1999}. Recently, the notions of meet and join in the majorization lattice have been given new interpretations in the context of Quantum Resource Theories (QRTs). Indeed, the meet was interpreted as an Optimal Common Resource~\cite{Bosyk2019}, while the join was described as an Optimal Common Product~\cite{Deside2024}. In view of that, exploring the lattice-theoretic properties of functions is natural in the context of QRTs.

\section{Majorization precursors to supermodularity and subadditivity on the lattice} \label{sec:precursors}

The Shannon entropy is known to be supermodular and subadditive on the majorization lattice, as proved by Cicalese and Vaccaro in Ref.~\cite{Cicalese2002}. The proof that they proposed is directly particularized to the case of the Shannon entropy, but one may ask whether it could also be due to an underlying majorization relation, what we call a \textit{majorization precursor}, which would be stronger than the entropic inequality (and would a priori imply supermodularity for a broader class of functions).
Such a majorization relation could relate tensor products of the form $(\bos{p} \wedge \bos{q}) \otimes (\bos{p} \vee \bos{q})$ and $\bos{p} \otimes \bos{q}$, and in turn imply inequalities relating sums of entropies of the form $H(\bos{p} \wedge \bos{q}) + H(\bos{p} \vee \bos{q})$ and $H(\bos{p}) + H(\bos{q})$.
Interestingly, we show that the operation of interest that allows us to tackle such a problem turns out to be a simple vector concatenation (or direct sum), rather than a tensor product of vectors. Such concatenations have previously appeared in the context of the study of, e.g., quantum uncertainty relations~\cite{Rudnicki2014, Rudnicki2015}.

We define the direct sum of two vectors $\bos{p}, \bos{q} \in \mathcal{P}_d$ as follows:
\begin{equation}
    \bos{p} \oplus \bos{q} \coloneqq (p_1, p_2, ..., p_d, q_1, q_2, ..., q_d)^\da.
\end{equation}
Note that $\bos{p} \oplus \bos{q} \in \mathbb{R}^{2d}$, and $S^\da_{2d}(\bos{p} \oplus \bos{q}) = 2$.

\begin{theo}\label{th:concatenation}
    For any $\bos{p}, \bos{q} \in \mathcal{P}_d$, we have
    \begin{equation}
         (\bos{p} \wedge \bos{q}) \oplus (\bos{p} \vee \bos{q}) \prec \bos{p} \oplus \bos{q}.
    \end{equation}
\end{theo}
\begin{proof} \label{proof:alternative_supermodularity}
    The proof goes in two steps. We define the following vectors: 
    \begin{gather}
        \bos{m} := \bos{p} \wedge \bos{q}, \quad \bos{j} := \bos{p} \vee \bos{q}, \quad \bos{\beta} := \bos{\beta}(\bos{p}, \bos{q}),\\
        \bos{A} := \bos{p} \oplus \bos{q} \in \mathbb{R}^{2d}, \label{eq:A}\\
        \bos{B} := \bos{m} \oplus \bos{j} \in \mathbb{R}^{2d}, \label{eq:B}\\
        \bos{B'} := \bos{m} \oplus \bos{\beta} \in \mathbb{R}^{2d}.
    \end{gather}
    The first step of the proof is to show that $\bos{B'} \prec \bos{A}$. The second is to show that since $\bos j \prec \bos{\beta}$, then $\bos{B} \prec \bos{B'}$ as well. Note that $\sum_{i=1}^{2d} A_i = \sum_{i=1}^{2d} B_i = \sum_{i=1}^{2d} B'_i = 2$.
    We begin by proving that $\bos{B'} \prec \bos{A}$, by showing that all of the terms in $\bos{B'}$ can be obtained by successive $T$-transforms of $\bos{A}$~\cite{Majorization}. More specifically, a $T$-transform is a matrix of the form $\lambda I + (1 - \lambda) P_{ij}$, where $0 \leq \lambda \leq 1$, $I$ is the identity matrix, and $P_{ij}$ is the permutation matrix swapping coordinates $i$ and $j$. Note that for any vector $\bos p \in \mathbb{R}^d_+$, $T \bos p^t$ is always  majorized by $\bos p^t$ \cite{Majorization}, where $\cdot^t$ denotes the transpose. Let us assume that for a given cumulative sum index $i \leq d-1$, we have $S^\da_{i}(\bos{p}) \geq S^\da_{i}(\bos{q})$ (the other case being perfectly symmetric). We have 2 possible cases for the index $i + 1$.
    \begin{enumerate}
        \item \underline{$S^\da_{i+1}(\bos{p}) \geq S^\da_{i+1}(\bos{q})$:} \label{case:simple} in this case, the expressions for $m_{i+1}$ and $\beta_{i+1}$ are simple. We have
        \begin{align}
            m_{i+1} = S^\da_{i+1}(\bos{q}) - S^\da_{i}(\bos{q}) = q_{i+1} \\
            \beta_{i+1} = S^\da_{i+1}(\bos{p}) - S^\da_{i}(\bos{p}) = p_{i+1},
        \end{align}
        and so, for the index $i+1$, the components of $\bos{m}$ and $\bos{\beta}$ are simply components of $\bos{p}$ and $\bos{q}$. Therefore, for all of the indices falling under this case, the components of $\bos{B'}$ are simply the same as those of $\bos{A}$.
        \item \underline{$S^\da_{i+1}(\bos{p}) < S^\da_{i+1}(\bos{q})$:} \label{case:complicated} in this case, the difference of cumulative sums changes sign from index $i$ to index $i+1$. Let us define $\Delta = S^\da_{i}(\bos{p}) - S^\da_{i}(\bos{q})$, which is greater than or equal to 0 by hypothesis. We have
        \begin{align}
            m_{i+1} = S^\da_{i+1}(\bos{p}) - S^\da_{i}(\bos{q}) = p_{i+1} + \Delta \label{eq:meet_terms}\\
            \beta_{i+1} = S^\da_{i+1}(\bos{q}) - S^\da_{i}(\bos{p}) = q_{i+1} - \Delta \label{eq:join_terms}.
        \end{align}
        Moreover, we know that $p_{i+1} + \Delta < q_{i+1}$, or equivalently, $q_{i+1} - \Delta > p_{i+1}$. Indeed,
        \begin{align}
            0 &< S^\da_{i+1}(\bos{q}) - S^\da_{i+1}(\bos{p})\\
            &= q_{i+1} - p_{i+1} + S^\da_{i}(\bos{q}) - S^\da_{i}(\bos{p})\\
            \Leftrightarrow \quad q_{i+1} - &p_{i+1} > \underbrace{S^\da_{i}(\bos{p}) - S^\da_{i}(\bos{q})}_{= \Delta} \geq 0.
        \end{align}
        Thus, for the components associated with index $i+1$, we find the following $T$-transform
        \begin{equation}
            T = \left(1-\frac{\Delta}{q_{i+1} - p_{i+1}}\right)I + \frac{\Delta}{q_{i+1} - p_{i+1}} P_{12},
        \end{equation}
        such that $(m_{i+1}, \beta_{i+1})^t = T(p_{i+1}, q_{i+1})^t$.
    \end{enumerate}
    Therefore, since cases~\ref{case:simple} and~\ref{case:complicated} give expressions of $T$-transforms for going from $\bos{A}$ to $\bos{B'}$ (being trivially the identity matrix in case~\ref{case:simple}), we find $\bos{B'} \prec \bos{A}$.

    Now, in order to finalize the proof, we need to show that $\bos{B} \prec \bos{B'}$. Clearly, the $k^{\text{th}}$ cumulative sum of $\bos{B}$ is made up of the largest entries of $\bos{m}$ and $\bos{j}$, and the same goes for cumulative sums of $\bos{B'}$ being made of the largest entries of $\bos{m}$ and $\bos{\beta}$. To avoid ill-defined terms, we additionally introduce the convention that $S^\da_k (\bos{v}) = 1, \; \forall k > d$, for any $d$-dimensional probability vector $\bos{v}$. We already know from Ref.~\cite{Cicalese2002} that $\bos j \prec \bos{\beta}$, so it only remains to prove that $\bos{j} \prec \bos{\beta} \Rightarrow \bos{j} \oplus \bos{m} \prec \bos{\beta} \oplus \bos{m}$. We have
    \begin{gather}
        S^\da_k (\bos{B}) = \max_{0\leq l\leq k} \left(S^\da_l (\bos{j}) + S^\da_{k-l} (\bos{m})\right), \label{eq:max_B} \\
        S^\da_k (\bos{B'}) = \max_{0\leq l'\leq k} \left(S^\da_{l'} (\bos{\beta}) + S^\da_{k-l'} (\bos{m})\right). \label{eq:max_B'}
    \end{gather}
    Let us call $L$ the first value of $l$ that realizes the maximum of Eq.~(\ref{eq:max_B}), and $L'$ the first value of $l'$ that realizes the maximum of Eq.~(\ref{eq:max_B'}). Since $\bos{j} \prec \bos{\beta}$, we know that $S^\da_i (\bos{j}) \leq S^\da_i (\bos{\beta})$ is true for all $i \leq 2d$. To prove $\bos{B} \prec \bos{B'}$, we need to show that for all $k \leq 2d$, $S^\da_k (\bos{B}) \leq S^\da_k (\bos{B'})$ is true. First, note that since $\bos j \prec \bos \beta$, $L' < L$ is not possible, because $S^\da_l (\bos{j}) + S^\da_{k-l} (\bos{m}) \leq S^\da_{l} (\bos{\beta}) + S^\da_{k-l} (\bos{m})$ for all $l$. Therefore, for any $k \leq 2d$, there are 2 possible cases.
    
    \begin{enumerate}
        \item \underline{$L' = L$:} we have
        \begin{equation}
            \begin{aligned}
                S^\da_k (\bos{B}) &= S^\da_L (\bos{j}) + S^\da_{k-L} (\bos{m})\\
            &\leq S^\da_L (\bos{\beta}) + S^\da_{k-L} (\bos{m}) = S^\da_k (\bos{B'}).
            \end{aligned}
        \end{equation}
        \item \underline{$L' > L$:} we have
        \begin{equation}
            \begin{aligned}
                S^\da_k (\bos{B'}) &= \max_{0\leq l'\leq k} \left(S^\da_{l'} (\bos{\beta}) + S^\da_{k-l'} (\bos{m})\right)\\
                & = S^\da_{L'} (\bos{\beta}) + S^\da_{k-L'} (\bos{m})\\
                &\geq S^\da_L (\bos{\beta}) + S^\da_{k-L} (\bos{m})\\
                &\geq S^\da_L (\bos{j}) + S^\da_{k-L} (\bos{m}) = S^\da_k (\bos{B}).
            \end{aligned}
        \end{equation}
    \end{enumerate}
    Hence, in both cases, $S^\da_k (\bos{B}) \leq S^\da_k (\bos{B'})$ is true. Therefore, $\bos{B} \prec \bos{B'}$, and so $\bos{B} \prec \bos{A}$ as well by transitivity of the majorization relation.
\end{proof}

The majorization precursor expressed by Theorem~\ref{th:concatenation} directly implies the supermodularity property for sum-concave functions. Note that for the rest of this Section, all statements about sum-concave functions can be mirrored for sum-convex functions by flipping the direction of the inequality, but we do not state the analog results for the sake of brevity.

\begin{cor}\label{cor:alternative_supermodularity}
    All sum-concave functions $F$ are supermodular on the majorization lattice. Namely, for any $\bos{p}, \bos{q} \in \mathcal{P}_d$, we have
    \begin{equation} \label{eq:supermodularity}
        F(\bos{p} \wedge \bos{q}) + F(\bos{p} \vee \bos{q}) \geq F(\bos{p}) + F(\bos{q}).
    \end{equation}
\end{cor}
\noindent Note that if $\bos{p} \sim \bos{q}$, this is a trivial equality.
\begin{proof}[Proof of Corollary~\ref{cor:alternative_supermodularity}]
    For $\bos{x} \in \mathbb{R}^{2d}$, let
    \begin{equation}
        F(\bos{x}) = \sum_i \varphi(x_i),
    \end{equation}
    with $\varphi$ concave.
    From Theorem~\ref{th:concatenation}, we have that $\bos{A} \succ \bos{B}$ (where $\bos{A}$ and $\bos{B}$ are defined in Eqs.~\eqref{eq:A} and~\eqref{eq:B}), and using Lemma~\ref{lem:karamata} on $\varphi$ ($-\varphi$ being convex on the same interval), we get
    \begin{equation} \label{eq:supermodularity_karamata_step}
        \sum_{i=1}^{2d} \varphi(A_i) \leq \sum_{i=1}^{2d} \varphi(B_i).
    \end{equation}
    By the sum nature of $F$, the LHS of Eq.~(\ref{eq:supermodularity_karamata_step}) is precisely $F(\bos{p}) + F(\bos{q})$, and the RHS is precisely $F(\bos{p} \wedge \bos{q}) + F(\bos{p} \vee \bos{q})$, so that we have proven Eq.~(\ref{eq:supermodularity}).
\end{proof}

In some sense, Corollary~\ref{cor:alternative_supermodularity} is an interesting parallel to Lemma~\ref{lem:karamata}: it has been known for a long time that functions of the form show in Eq.~\eqref{eq:sum-cc} are the simplest form of Schur-concave functions, but we now show that on the majorization lattice, they can also be understood as the simplest form of supermodular functions. That this can be traced back to the meet/join concatenation being more spread out than the concatenation of two probability distributions (as is shown by the majorization relation of Theorem~\ref{th:concatenation}) gives useful insight on the interaction between these functions and the lattice-theoretic operations.

Corollary~\ref{cor:alternative_supermodularity} and its analogue for submodular functions essentially show that for the class of sum-convex/concave functions, submodularity/supermodularity is not so much a specificity of the functions themselves, but rather a consequence of the structure of the majorization lattice. As an example, the functions $||\bos{x}||_\alpha^\alpha = \sum_i x_i^\alpha$ are submodular for $\alpha > 1$ and supermodular for $\alpha < 1$, which we expect to have interesting applications in Quantum Resource Theories using the results from~\cite{Aubrun2008} on the $\ell_1$ closure of catalytic sets.
We postpone further application of these results to entropies to Section~\ref{sec:entropies}.

In the same spirit, one may ask whether a majorization precursor for subadditivity can be shown. The following two results answer precisely this question. Note that for the remainder of this work, we denote the distribution $(1, 0, \dots, 0)$ by the symbol $\bos{e}$ (which thus depends on the dimension of the vectors involved).
\begin{lem}\label{lem:subadditive_tensor}
    For any $\bos{p}, \bos{q} \in \mathcal{P}_d$, we have
    \begin{equation}
        \bos{p} \otimes \bos{q} \prec \bos{p} \wedge \bos{q}.
    \end{equation}
\end{lem}

\begin{proof}
    Let $\bos{p}, \bos{q} \in \pr$. It is immediate to see that $\bos{p} \otimes \bos{q} \prec \bos{p}$ and $\bos{p} \otimes \bos{q} \prec \bos{q}$ are both always true, which implies that $\bos{p} \otimes \bos{q} \prec \bos{p} \wedge \bos{q}$ is also true by definition of the meet.
\end{proof}

\begin{theo}\label{th:subadditive_concatenation}
    For any $\bos{p}, \bos{q} \in \mathcal{P}_d$, we have
    \begin{equation}
        \bos{p} \oplus \bos{q} \prec (\bos{p} \wedge \bos{q}) \oplus \bos e,
    \end{equation}
    where $\bos{e} = (1, 0, \dots, 0)$.
\end{theo}

\begin{proof}
            Let $\bos{p}, \bos{q} \in \mathcal{P}_d$. We define 
    \begin{gather}
        \bos{m} := \bos{p} \wedge \bos{q},\\
        \bos{A} := \bos{p} \oplus \bos{q},\label{eq:A2}\\
        \bos{C} := \bos{m} \oplus \bos{e}.\label{eq:C}
    \end{gather}
    Note that $\sum_{i=1}^{2d} A_i = \sum_{i=1}^{2d} C_i = 2$. We now show that the majorization precursor $\bos{A} \prec \bos{C}$ holds. Clearly, the $k^{\text{th}}$ cumulative sum of $\bos{A}$ is made up of the largest entries of $\bos{p}$ and $\bos{q}$, and the same goes for the cumulative sums of $\bos{C}$ being made of the largest entries of $\bos{m}$ and $\bos{e}$. To avoid ill-defined terms, we additionally introduce the convention that, for any $\bos{v} \in \pr$, $S^\da_k (\bos{v}) = 1, \; \forall k > d$, and $S^\da_0(\bos{v}) = 0$. For all $k \leq 2d$, we have
    \begin{gather}
        S^\da_k (\bos{A}) = \max_{0\leq l\leq k} \left(S^\da_l (\bos{p}) + S^\da_{k-l} (\bos{q})\right), \label{eq:theo 2 Ska}\\
        S^\da_k (\bos{C}) = 1 + S^\da_{k-1}(\bos{m}).
    \end{gather} 
    Recall Eq.~(\ref{eq:alpha_bis}), from which we immediately deduce
    \begin{equation}
        S^\da_i (\bos{m}) = \min \left\{S^\da_i (\bos{p}), S^\da_i (\bos{q})\right\}.
    \end{equation}
    We need to show that $S^\da_k (\bos{A}) \prec S^\da_k (\bos{C})$ is true for all $k \leq 2d$. First, we directly see that, when $k = 1$, we have $S^\da_1(\bos{C}) = 1 \geq S^\da_1(\bos{A}) = \max \{p_1, q_1\}$. 
    Now, for $k \geq 2$, there are several cases depending on the value of $L_k$, i.e. the integer value at which the maximum in Eq.~\eqref{eq:theo 2 Ska} is realized, which ranges from $0$ to $k$. The vectors $\bos{p}$ and $\bos{q}$ having no particular structure, we must consider several cases separately. 
    \begin{enumerate}
        \item \underline{$L_k = 0$}: we have \label{case:0}
        \begin{equation}
            S^\da_k(\bos{C}) \geq 1 \geq S^\da_k(\bos{q})=S^\da_k(\bos{A}).
        \end{equation}
        \item \underline{$L_k = k$}: symmetric to the $L_k=0$ case, because \label{case:k}
        \begin{equation}
            S^\da_k(\bos{C}) \geq 1 \geq S^\da_k(\bos{p})=S^\da_k(\bos{A}).
        \end{equation}
    \end{enumerate}
    For the remaining cases, we first note the following inequality
    \begin{align}
        S^\da_k(\bos{C}) &= 1 + \min \left\{S^\da_{k-1}(\bos{p}), S^\da_{k-1}(\bos{q})\right\}\\
        &\geq \max \left\{S^\da_{k-1}(\bos{p}), S^\da_{k-1}(\bos{q})\right\}\nonumber\\
        &\quad + \min \left\{S^\da_{k-1}(\bos{p}), S^\da_{k-1}(\bos{q})\right\} \label{eq:subadditivity_max_trick}\\
        &= S^\da_{k-1}(\bos{p}) + S^\da_{k-1}(\bos{q}),
    \end{align}
    where Eq.~(\ref{eq:subadditivity_max_trick})~is obtained using that the cumulative sum of any $d$-dimensional probability vector is always less than or equal to 1. Therefore, for the relation $S^\da_k(\bos{C}) \geq S^\da_k(\bos{A})$ to be verified, it is enough to show that, for any choice of $L_k$ that we have not treated yet,
    \begin{gather}
        S^\da_{k-1}(\bos{p}) + S^\da_{k-1}(\bos{q}) \geq \left(S^\da_{L_k} (\bos{p}) + S^\da_{k-{L_k}} (\bos{q})\right) = S^\da_k(\bos{A}),\label{eq:subadditivity_leq}
    \end{gather}
    for all $k \geq 2$.
    \begin{enumerate}[resume]
        \item \underline{$L_k = 1$}: Eq.~(\ref{eq:subadditivity_leq}) becomes \label{case:1}
        \begin{equation}
            S^\da_1(\bos{p}) - S^\da_{k-1}(\bos{p}) \leq 0,
        \end{equation}
        which is verified for all $k \geq 2$.
        \item \underline{$L_k = k - 1$}: symmetric to the $L_k = 1$ case, because Eq.~(\ref{eq:subadditivity_leq}) becomes \label{case:k-1}
        \begin{equation}
            S^\da_1(\bos{q}) - S^\da_{k-1}(\bos{q}) \leq 0,
        \end{equation}
        which is verified for all $k \geq 2$.
        \item \underline{$2 \leq L_k \leq k-2$}: this case can only arise for $k \geq 4$, otherwise the inequality allows no $L_k$ value\footnote{This is not an issue because if $k=2$, $L_2 \in \{0, 1, 2\}$, and if $k = 3$, $L_3 \in \{0, 1, 2, 3\}$. All of these possibilities fall under cases~\ref{case:0},~\ref{case:k},~\ref{case:1} or~\ref{case:k-1}, thus all possible cases are taken into account.}. We have both \label{case:rest}
        \begin{equation}
            S^\da_l (\bos{p}) \leq S^\da_{k-1}(\bos{p}) \quad \text{and} \quad S^\da_{k-l} (\bos{q}) \leq S^\da_{k-1}(\bos{q}),
        \end{equation}
        which are both true for all $k \geq 4$, and so Eq.~(\ref{eq:subadditivity_leq}) is verified.
    \end{enumerate}
    Therefore, no matter the value of $L_k$, the majorization inequalities are satisfied for all $k \leq 2d$. We have thus shown that $\bos{A} \prec \bos{C}$.
\end{proof}

The majorization precursor expressed by Theorem~\ref{th:subadditive_concatenation} directly implies the following property for sum-concave functions.
\begin{cor}\label{cor:subadditivity}
    For any sum-concave function $F$ and any $\bos{p}, \bos{q} \in \mathcal{P}_d$, we have
    \begin{equation} \label{eq:subadditivity}
        F(\bos{p} \wedge \bos{q}) \leq F(\bos{p}) + F(\bos{q}) - F(\bos e),
    \end{equation}
    where $\bos e = (1, 0, \dots, 0)$.
\end{cor}
\noindent Interestingly, for functions such that $F(\bos e) > 0$, this is a stronger property than the usual subadditivity on the majorization lattice.
\begin{proof}[Proof of Corollary~\ref{cor:subadditivity}]
    For $\bos{x} \in \mathbb{R}^{2d}$, let
    \begin{equation}
        F(\bos{x}) = \sum_i \varphi(x_i),
    \end{equation}
    with $\varphi$ concave.
    From Theorem~\ref{th:subadditive_concatenation} we have that $\bos{A} \prec \bos{C}$ (where $\bos{A}$ and $\bos{C}$ are defined in Eqs.~\eqref{eq:A2} and~\eqref{eq:C}), and using Lemma~\ref{lem:karamata} on $\varphi$ ($-\varphi$ being convex on the same interval), we get
    \begin{equation} \label{eq:subadditivity_karamata_step}
        \sum_{i=1}^{2d} \varphi(A_i) \geq \sum_{i=1}^{2d} \varphi(C_i).
    \end{equation}
    By the sum nature of $F$, the LHS of Eq.~(\ref{eq:subadditivity_karamata_step}) is precisely $F(\bos{p}) + F(\bos{q})$, and the RHS is precisely $F(\bos{p} \wedge \bos{q}) + F(\bos{e})$, so we have proven Eq.~(\ref{eq:subadditivity}). 
\end{proof}

\section{Refined supermodularity and subadditivity of Shannon, R\'enyi, and Tsallis entropies} \label{sec:entropies}

This section proves the supermodularity and subadditivity of Tsallis entropies using Theorems~\ref{th:concatenation},~\ref{th:subadditive_concatenation}, and Lemma~\ref{lem:subadditive_tensor}, recovering these properties for the Shannon entropy.
We then combine the majorization precursors with the improved Schur-concavity of Shannon and Tsallis entropies proven by Ho and Verd\'u in Ref.~\cite{HoVerdu2010, HoVerdu2015} to further tighten the bounds on supermodularity and subadditivity, essentially showing \textit{strict} supermodularity and subadditivity. Moreover, we prove the supermodularity (for $\alpha > 1$) and subadditivity (for $\alpha \geq 0$) of R\'enyi entropies on the majorization lattice using an additional lemma on log-submodularity. Let us begin with the Shannon entropy. To prove its supermodularity, one could simply notice that it is sum-concave (with $\varphi(x) = -x \log x$) and use Corollary~\ref{cor:alternative_supermodularity}, however one can achieve a better bound with a slightly different approach. We will use the following lemma.
\begin{lem}[Ho and Verd\'u~\cite{HoVerdu2010}] \label{lem:ho_verdu_shannon}
    Let $\bos{p}, \bos{q} \in \pr$ be such that $\bos{p} \prec \bos{q}$. Then,
    \begin{equation} \label{eq:ho_verdu_shannon}
        H(\bos{p}) \geq H(\bos{q}) + D(\bos{q} \parallel \bos{p}).
    \end{equation}
\end{lem}

Using the above lemma, we can show the following strengthening of the supermodularity of the Shannon entropy on the majorization lattice.
\begin{cor}\label{cor:refined_shannon_supermodularity}
    For any $\bos{p}, \bos{q} \in \mathcal{P}_d$, we have
    \begin{equation} \label{eq:refined_shannon_supermodularity}
        H(\bos{p} \wedge \bos{q}) + H(\bos{p} \vee \bos{q}) \geq H(\bos{p}) + H(\bos{q}) + \eta,
    \end{equation}
    where $\eta \coloneqq 2D\left(\tfrac{\bos{p} \oplus \bos{q}}{2} \parallel \tfrac{\bos{p} \wedge \bos{q} \oplus \bos{p} \vee \bos{q}}{2}\right)$.
\end{cor}

\begin{proof}
    From Theorem~\ref{th:concatenation} we have that $\tfrac{\bos{p} \oplus \bos{q}}{2} \succ \tfrac{\bos{m} \oplus \bos{j}}{2}$, and using Lemma~\ref{lem:ho_verdu_shannon}, we directly obtain
    \begin{align}
        &H\left(\tfrac{\bos{m} \oplus \bos{j}}{2}\right) \geq H\left(\tfrac{\bos{p} \oplus \bos{q}}{2}\right) + D\left(\tfrac{\bos{p} \oplus \bos{q}}{2} \parallel \tfrac{\bos{m} \oplus \bos{j}}{2}\right),\\
        &\Rightarrow H(\bos{p} \wedge \bos{q}) + H(\bos{p} \vee \bos{q})\nonumber\\
        &\quad \quad\geq H(\bos{p}) + H(\bos{q}) + 2 D\left(\tfrac{\bos{p} \oplus \bos{q}}{2} \parallel \tfrac{\bos{m} \oplus \bos{j}}{2}\right),
    \end{align}
    and so the supermodularity of the Shannon entropy on the majorization lattice can be improved by the term $2 D\left(\tfrac{\bos{p} \oplus \bos{q}}{2} \parallel \tfrac{\bos{m} \oplus \bos{j}}{2}\right)$.
\end{proof}

Since $\bos{p} \oplus \bos{q} \neq \bos{p} \wedge \bos{q} \oplus \bos{p} \vee \bos{q}$ whenever $\bos{p} \nsim \bos{q}$ and $D(\bos{x}\parallel \bos{y}) > 0$ whenever $\bos{x} \neq \bos{y}$, Corollary~\ref{cor:refined_shannon_supermodularity} implies that, on the majorization lattice, the Shannon entropy is \textit{strictly} supermodular, i.e. the inequality that Cicalese and Vaccaro proved in Ref.~\cite{Cicalese2002} can \textit{never} be saturated for incomparable vectors.

A very similar approach using Lemma~\ref{lem:subadditive_tensor} and Theorem~\ref{th:subadditive_concatenation} allows one to also get a better bound on the subadditivity of the Shannon entropy on the majorization lattice, which is thus \textit{strict}.
\begin{cor}
    For any $\bos{p}, \bos{q} \in \mathcal{P}_d$, we have
    \begin{equation} \label{eq:refined_shannon_subadditivity}
        H(\bos{p}) + H(\bos{q}) \geq H(\bos{p} \wedge \bos{q}) + \eta,
    \end{equation}
    where $\eta = \max\left\{D\left(\bos{p} \wedge \bos{q} \parallel \bos{p} \otimes \bos{q}\right), 2D\left(\frac{\bos{p} \wedge \bos{q} \oplus \bos{e}}{2} \parallel \frac{\bos{p} \oplus \bos{q}}{2}\right)\right\}$, and $\bos e = (1, 0, \dots, 0)$.
\end{cor}

Another consequence of Corollaries~\ref{cor:alternative_supermodularity} and~\ref{cor:subadditivity} is the application to Tsallis entropies, whose supermodularity and subadditivity on the majorization lattice can be proven. Again, we obtain better bounds using a similar approach to the Shannon case. We will use the following lemma.
\begin{lem}[Ho and Verd\'u~\cite{HoVerdu2015}] \label{lem:ho_verdu_tsallis}
     Define $\phi_\alpha(x) = \frac{x^\alpha - x}{1 - \alpha}$ and $\Delta_{\phi_\alpha}(x, y) = (x - y)\phi'_\alpha(x) + \phi_\alpha(y) - \phi_\alpha(x)$, and let $\bos{p}, \bos{q} \in \pr$ be such that $\bos{p} \prec \bos{q}$. Then,
    \begin{equation} \label{eq:ho_verdu_tsallis}
        T_\alpha(\bos{p}) \geq T_\alpha(\bos{q}) + W_{\phi_\alpha}(\bos{q} \parallel \bos{p}),
    \end{equation}
    where $W_{\phi_\alpha}(\bos{q} \parallel \bos{p}) = \sum_i \Delta_{\phi_\alpha}(q_i, p_i)$.
\end{lem}

We can now show the strict supermodularity and strict subadditivity of Tsallis entropies on the majorization lattice.
\begin{cor} \label{cor:refined_tsallis_supermodularity}
    For any $\bos{p}, \bos{q} \in \mathcal{P}_d$ and any $\alpha \in \mathbb{R}^+$,
    \begin{equation}
        T_\alpha(\bos{p} \wedge \bos{q}) + T_\alpha(\bos{p} \vee \bos{q}) \geq T_\alpha(\bos{p}) + T_\alpha(\bos{q}) + \tau,
    \end{equation}
    where $\tau = 2^\alpha W_{\phi_\alpha}(\tfrac{\bos{p} \oplus \bos{q}}{2}\parallel\tfrac{\bos{p} \wedge \bos{q} \oplus \bos{p} \vee \bos{q}}{2})$.
\end{cor}

\begin{proof}
     From Theorem~\ref{th:concatenation} we have that $\tfrac{\bos{p} \oplus \bos{q}}{2} \succ \tfrac{\bos{m} \oplus \bos{j}}{2}$, and using Lemma~\ref{lem:ho_verdu_tsallis} we directly obtain
     \begin{equation}
         \begin{aligned}
             & T_\alpha\left(\tfrac{\bos{m} \oplus \bos{j}}{2}\right) \geq T_\alpha\left(\tfrac{\bos{p} \oplus \bos{q}}{2}\right) + W_{\phi_\alpha}\left(\tfrac{\bos{p} \oplus \bos{q}}{2}\parallel \tfrac{\bos{m} \oplus \bos{j}}{2}\right),\\
             \Rightarrow \quad & T_\alpha(\bos{p} \wedge \bos{q}) + T_\alpha(\bos{p} \vee \bos{q}) \\
             & \quad \quad \geq T_\alpha(\bos{p}) + T_\alpha(\bos{q}) + 2^\alpha W_{\phi_\alpha}\left(\tfrac{\bos{p} \oplus \bos{q}}{2}\parallel \tfrac{\bos{m} \oplus \bos{j}}{2}\right),
         \end{aligned}
     \end{equation}

     and so the supermodularity of Tsallis entropies on the majorization lattice is improved by the term $2^\alpha W_{\phi_\alpha}\left(\tfrac{\bos{p} \oplus \bos{q}}{2}\parallel \tfrac{\bos{m} \oplus \bos{j}}{2}\right)$.
\end{proof}

Since $\bos{p} \oplus \bos{q} \neq \bos{p} \wedge \bos{q} \oplus \bos{p} \vee \bos{q}$ whenever $\bos{p} \nsim \bos{q}$ and $W_{\phi_\alpha}(\bos{x}\parallel \bos{y}) > 0$ whenever $\bos{x} \neq \bos{y}$~\cite{HoVerdu2015}, Corollary~\ref{cor:refined_shannon_supermodularity} thus implies that the Tsallis entropy is \textit{strictly} supermodular on the majorization lattice.

Just like for the Shannon entropy, a similar approach using Lemma~\ref{lem:subadditive_tensor} and Theorem~\ref{th:subadditive_concatenation} allows one to get a better bound on the subadditivity of the Tsallis entropy on the majorization lattice.
\begin{cor} \label{cor:refined_tsalllis_subadditivity}
    For any $\bos{p}, \bos{q} \in \mathcal{P}_d$ and any $\alpha \in \mathbb{R}^+$,
    \begin{equation}
        T_\alpha(\bos{p}) + T_\alpha(\bos{q}) \geq T_\alpha(\bos{p} \wedge \bos{q}) + \tau,
    \end{equation}
    where $\tau = 2^\alpha W_{\phi_\alpha}\left(\tfrac{\bos{m} \oplus \bos e}{2}\parallel \tfrac{\bos{p} \oplus \bos{q}}{2}\right)$, and $\bos e = (1, 0, \dots, 0)$.
\end{cor}

Tsallis entropies are not additive over a tensor product in general, so one cannot use Lemma~\ref{lem:subadditive_tensor}, contrarily to the Shannon case. Similarly, Corollary~\ref{cor:refined_tsalllis_subadditivity} implies that Tsallis entropies are \textit{strictly} subadditive on the majorization lattice.

Finally, we can proceed in similar ways for the R\'enyi entropy. The two following lemmas will be useful.
\begin{lem}[Ho and Verd\'u~\cite{HoVerdu2015}] \label{lem:ho_verdu_renyi}
    Let $\bos{p}, \bos{q} \in \mathcal{P}_d$ be such that $\bos{p} \prec \bos{q}$, and define $\phi_\alpha$ and $\Delta_{\phi_\alpha}$ as in Lemma~\ref{lem:ho_verdu_tsallis}. Then,
    \begin{equation} \label{eq:ho_verdu_renyi}
        H_\alpha(\bos{p}) \geq H_\alpha(\bos{q}) + \log_2(e) W_{\phi_\alpha}(\bos{q} \parallel \bos{p}),
    \end{equation}
    where $W_{\phi_\alpha}(\bos{q} \parallel \bos{p}) = \sum_i \Delta_{\phi_\alpha}(q_i, p_i)$, and where the $\log_2(e)$ comes from the conversion from nats to bits.
\end{lem}

\begin{lem}[{Topkis~\cite{Topkis1998}}] \label{lem:log-submodularity}
    If a function $F$ is increasing\footnote{Increasing is to be understood as $\bos{p} \prec \bos{q} \Rightarrow F(\bos{p}) \leq F(\bos{q})$, which corresponds to Schur-convexity for the majorization partial ordering.} (or decreasing) and submodular on a lattice, then it is also log-submodular, meaning that we have
    \begin{equation}
        F(\bos{p}) F(\bos{q}) \geq F(\bos{p} \wedge \bos{q}) F(\bos{p} \vee \bos{q}).
    \end{equation}
\end{lem}

This lemma is quite useful, as it allows us to bridge the gap between the majorization relation on direct sums of distributions of Theorem~\ref{th:concatenation} back to sums of logarithms, which we need for R\'enyi entropies.

\begin{cor}\label{cor:renyi_supermodularity}
    R\'enyi entropies of $\alpha > 1$ are supermodular on the majorization lattice, and so for any $\bos{p}, \bos{q} \in \mathcal{P}_d$  and any $\alpha \in (1, \infty)$,
    \begin{equation}
        H_\alpha(\bos{p}) + H_\alpha(\bos{q}) \leq H_\alpha(\bos{p} \wedge \bos{q}) + H_\alpha (\bos{p} \vee \bos{q}).
    \end{equation}
\end{cor}

\begin{proof}
    Consider the function $F(\bos{p}) = ||\bos{p}||_\alpha^\alpha = \sum_i p_i^\alpha$. Since the function $f(x) = x^\alpha$ is convex for $x \in [0, 1], \alpha \in (1, +\infty)$, $F(\bos{p})$ is clearly sum-convex for any $\alpha > 1$. By Corollary~\ref{cor:alternative_supermodularity} applied to $(-F)$, $F$ is thus submodular on the majorization lattice, and by Lemma~\ref{lem:log-submodularity} we have that for any pair of probability distributions $\bos{p}, \bos{q} \in \mathcal{P}_d$ and any $\alpha > 1$,
    \begin{align}
        ||\bos{p}||_\alpha^\alpha||\bos{q}||_\alpha^\alpha &\geq ||\bos{p} \wedge \bos{q}||_\alpha^\alpha||\bos{p} \vee \bos{q}||_\alpha^\alpha, \label{eq:log-submodularity}\\
        \Rightarrow \quad H_\alpha(\bos{p}) + H_\alpha(\bos{q}) &\leq H_\alpha(\bos{p} \wedge \bos{q}) + H_\alpha (\bos{p} \vee \bos{q}) \label{eq:renyi_supermodularity},
    \end{align}
    where we have used the fact that $H_\alpha(\bos{p}) = \frac{1}{1-\alpha}\log(||\bos{p}||_\alpha^\alpha)$, and so all R\'enyi entropies for $\alpha > 1$ are supermodular on the majorization lattice.
\end{proof}

This result is, by itself, quite interesting, as it shows that supermodularity on the majorization lattice is not merely a property enjoyed by the Shannon entropy, but also by a much broader class of entropies which are widely used in physics. One could hope that supermodularity also holds for the $\alpha < 1$ R\'enyi entropies, however numerical counterexamples exist. This might seem surprising at first given that one would expect a result analoguous to Lemma~\ref{lem:log-submodularity} to hold for supermodularity as well, however in that case it turns out that the implication goes the opposite way: log-supermodularity implies supermodularity, but not the other way around, and so supermodularity of the function $||\bos{p}||_\alpha^\alpha$ for $\alpha < 1$ (which holds by Corollary~\ref{cor:alternative_supermodularity}) does not imply that log-supermodularity also holds. It is interesting to note that supermodularity is the only non-trivial property of the Shannon entropy needed in Ref.~\cite{Cicalese2013} to show that the function $d_H(\bos{p}, \bos{q}) = H(\bos{p}) + H(\bos{q}) - 2 H(\bos{p} \vee \bos{q})$ is a distance on the majorization lattice. Therefore, Corollary~\ref{cor:renyi_supermodularity} implies that one can define a similar family of distances on the majorization lattice using R\'enyi entropies with $\alpha > 1$.

It is interesting to note that a very similar reasoning to that of the proof of Lemma~\ref{lem:log-submodularity} also holds for an analoguous property, namely that of log-subadditivity. Subadditivity holds for the $||\bos{x}||_\alpha^\alpha$ functions for $\alpha < 1$ for which we could use Corollary~\ref{cor:subadditivity} and log-subadditivity to show that R\'enyi entropies are also subadditive on the majorization lattice for $\alpha < 1$. However, for R\'enyi entropies, the majorization precursor of Theorem~\ref{lem:subadditive_tensor} (which works for all $\alpha$) coupled with Lemma~\ref{lem:ho_verdu_renyi} gives a better bound on subadditivity, as it is straightforward to show the following.

\begin{cor}\label{cor:renyi_subadditivity}
    For any $\bos{p}, \bos{q} \in \mathcal{P}_d$, and any $\alpha \in \mathbb{R}^+$,
    \begin{equation}
        H_\alpha(\bos{p}) + H_\alpha(\bos{q}) \geq H_\alpha(\bos{p} \wedge \bos{q}) + \eta,
    \end{equation}
    where $\eta = \log_2(e) W_{\phi_\alpha}(\bos{p} \wedge \bos{q} \parallel \bos{p} \otimes \bos{q})$, and so R\'enyi entropies are \textit{strictly} subadditive on the majorization lattice.
\end{cor}

Therefore, contrarily to the Tsallis case, Corollaries~\ref{cor:renyi_supermodularity} and~\ref{cor:renyi_subadditivity} suggest that the Shannon entropy (R\'enyi entropy for $\alpha \rightarrow 1$) is the ``last" function of the R\'enyi family of entropies that satisfies both the supermodularity ($\alpha > 1$) and subadditivity ($\alpha \geq 0$) properties. While we do not prove explicitly that no value of $\alpha$ below 1 allows for supermodularity ($\alpha = 0$ being a trivial case), there are plenty of numerical counterexamples. Note that this does not mean that submodularity is satisfied either, but that in general no conclusion can be drawn for the value of $H_\alpha(\bos m) + H_\alpha(\bos j) - H_\alpha(\bos p) - H_\alpha(\bos q)$, as it can be greater than or less than 0 for $\alpha \in (0, 1)$.

\section{Conclusion} \label{sec:conclusion}

In this work, we have shown three majorization precursors to supermodularity and subadditivity on the majorization lattice. Moreover, we have shown that the family of sum-convex/sum-concave functions, which have long been known to be the simplest form of Schur-convex/concave functions, are also all submodular/supermodular on the majorization lattice. This opens the door for supermodular optimization schemes on the majorization lattice for the whole class of sum-concave functions, which are fundamental functions in majorization theory.

We have also combined our precursors with refinements on the Schur-concavity of Shannon, Tsallis and R\'enyi entropies from~\cite{HoVerdu2010, HoVerdu2015} to prove better bounds on the supermodularity of Shannon and Tsallis entropies on the majorization lattice, and on the subadditivity of Shannon, R\'enyi and Tsallis entropies on the majorization lattice, essentially showing that in all of those cases, the inequality can be replaced with a strict inequality. We also provide an expression for the correction term in all of these cases. In particular, the strict supermodularity of Shannon and Tsallis entropies on the majorization lattice is of interest. Many results regarding the optimization of supermodular functions are better behaved under strict supermodularity. Notably, the set of argmax of a strictly supermodular function forms a chain (i.e. a totally ordered set~\cite{Davey2002}), instead of only being a sublattice (see Theorems 2.7.1 and 2.7.5 in Ref.~\cite{Topkis1998}).

These properties are also of interest in quantum information, since the Shannon entropy and R\'enyi entropies are often useful resource quantifiers. Indeed, the Shannon entropy of Schmidt coefficients fully characterizes asymptotic convertibility~\cite{Bennett1996_bis} of entangled states, while R\'enyi entropies fully describe trumping majorization and thus entanglement catalysis~\cite{Turgut2007, Klimesh2007}.
We thus expect our new relations on the submodularity and supermodularity of entropies to play an interesting role for the characterization of entanglement transformations.

\section*{Note added}

During completion of this work, we became aware of two independent recent works where similar results have been obtained regarding the subadditivity and supermodularity of R\'enyi entropies on the majorization lattice, albeit using  different methods~\cite{Yadav2026, Bruno2026}.

\section*{Acknowledgments}

A.S. acknowledges support from the Université libre de Bruxelles (Belgium) under the Fonds de promotion du doctorat.
M.G.J. acknowledges funding from l’Agence Nationale de la Recherche (ANR, France) under project ANR-25‑CE47‑4015.
S.D. is a FRIA grantee of the Fonds de la Recherche Scientifique – FNRS (Belgium).
N.J.C. acknowledges support by the Fonds de la Recherche Scientifique – FNRS (Belgium) under project CHEQS within the Excellence of Science (EOS) program.

\bibliographystyle{IEEEtran}
\bibliography{info}

\end{document}